\documentclass[11pt]{article}

\usepackage{amsmath,amssymb,amsbsy,amsfonts,amsthm,latexsym,
            amsopn,amstext,amsxtra,euscript,amscd,color}
\usepackage{verbatim}
\usepackage{float}

\usepackage{caption}

\newfont{\teneufm}{eufm10}
\newfont{\seveneufm}{eufm7}
\newfont{\fiveeufm}{eufm5}

\newfam\eufmfam
              \textfont\eufmfam=\teneufm \scriptfont\eufmfam=\seveneufm
              \scriptscriptfont\eufmfam=\fiveeufm

\newcommand{\fun}[3]{{{#1}\,:\,{#2}\,\rightarrow\,{#3}}}

\newtheorem{thm}{Theorem}
\newtheorem{lem}[thm]{Lemma}
\newtheorem{cor}[thm]{Corollary}

\newtheorem{prop}[thm]{Proposition}

\newtheorem{defn}[thm]{Definition}

\newtheorem{example}[thm]{Example}

\newcommand{\Tr}{{\rm Tr_n}}

\def\+{\oplus}

\def\cA{{\mathcal A}}
\def\cB{{\mathcal B}}
\def\cC{{\mathcal C}}

\def\cH{{\mathcal H}}

\def\cN{{\mathcal N}}

\def\cP{{\mathcal P}}

\def\cT{{\mathcal T}}
\def\cU{{\mathcal U}}

\def\cW{{\mathcal W}}

\def\C{{\mathbb C}}

\def\F{{\mathbb F}}
\def\Z{{\mathbb Z}}
\def\N{{\mathbb N}}

\def\F{{\mathbb F}}

\def\Z{{\mathbb Z}}

\def\V{{\mathbb V}}

\def\cGB{\mathcal{GB}}
\def\wt{{\rm wt}}

\newcommand{\zetak}[1][2^k]{\zeta_{{#1}}}

\def\aa{{\bf a}}
\def\bb{{\bf b}}
\def\cc{{\bf c}}
\def\dd{{\bf d}}
\def\uu{{\bf u}}
\def\vv{{\bf v}}
\def\ww{{\bf w}}
\def\xx{{\mathbf x}}

\def\yy{{\mathbf y}}
\def\zz{{\mathbf z}}

\def\00{{\bf 0}}
\def\11{{\bf 1}}

\def\\{\cr}
\def\({\left(}
\def\){\right)}

\newcommand{\BBR}{\mathbb{R}}

\newcommand{\BBC}{\mathbb{C}}

\newcommand{\BBF}{\mathbb{F}}

\def\wt{{\rm wt}}

\providecommand{\newoperator}[3]{%
  \newcommand*{#1}{\mathop{#2}#3}}

\newoperator{\FD}{\mathrm{FD}}{\nolimits}

\begin{document}

\title{\bf  Root-Hadamard transforms and complementary sequences\footnote{Dedicated to the memory of our friend and co-author, Francis N. Castro}}
\author{\Large Luis A. Medina$^1$, Matthew G. Parker$^2$,  Constanza Riera$^3$, \and \Large Pantelimon~St\u anic\u a$^4$
\vspace{0.4cm} \\
$^1$ Department of Mathematics, University of Puerto Rico,\\
 San Juan, PR 00925, USA; {\tt luis.medina17@upr.edu}\\
$^2$Department of Informatics, University of Bergen, \\
Bergen, Norway; {\tt Matthew.Parker@ii.uib.no}\\
$^3$Department of Computing, Mathematics, and Physics,\\
   Western Norway University of Applied Sciences \\
  5020 Bergen, Norway; {\tt csr@hvl.no}\\
$^4$Department of Applied Mathematics,
Naval Postgraduate School, \\
Monterey, CA 93943--5216,  USA; {\tt pstanica@nps.edu}
}
\date{\today}
\maketitle

\begin{abstract}
In this paper we define a new transform on (generalized) Boolean functions, which generalizes the Walsh-Hadamard, nega-Hadamard, $2^k$-Hadamard, consta-Hadamard and all $HN$-transforms. We describe the behavior of what we call the root- Hadamard transform   for a generalized Boolean function $f$ in terms of the binary components of $f$. Further, we define a notion of complementarity (in the spirit of the Golay sequences) with respect to this transform and furthermore, we describe the complementarity of a generalized Boolean set with respect to the binary components of the elements of that set.
\end{abstract}

\section{(Generalized) Boolean functions}

Let $\BBF_2^n$  be an $n$-dimensional vector space over the binary field $\F_2$ (here, we consider both $\F_2^n$ and the finite field $\F_{2^n}$) and, for an integer $q$, let $\Z_q$ be the ring of integers modulo $q$. By `$+$' and `$-$' we respectively denote addition
and subtraction in the appropriate environment (it will be obvious from the context).

 A function $F : \BBF_2^n \rightarrow \F_2$, $n>0$,
is called  a {\em  Boolean function} in $n$ variables, whose set will be denoted by  $\mathcal{B}_n$.
A Boolean function can be regarded as a multivariate polynomial over $\BBF_{2}$, called the {\em algebraic normal form}~(ANF)
\[
f(x_{1},\ldots,x_{n})=a_{0}+ \sum_{1\leq i\leq n}a_{i}x_{i}+ \sum_{1\leq i<j\leq n}a_{ij}x_{i}x_{j}+ \cdots+  a_{12\ldots n}x_{1}x_{2}\ldots x_{n},
\]
where the coefficients $a_{0},a_{ij},\ldots,a_{12\ldots n}\in\BBF_{2}$.
 The maximum number of variables in a monomial is called the ({\em algebraic}) {\em degree}. The ({\em Hamming}) {\em weight} of $\xx =(x_1,\ldots,x_n)\in \BBF_2^n$ is denoted by
$wt(\xx) $ and equals $ \sum_{i=1}^{n}x_i$
(the Hamming weight of a function is the weight of its truth table, that is, the weight of its output vector).
The cardinality of a set $S$ is denoted by $|S|$.

We order $\BBF_2^n$ lexicographically, and denote $\vv_0=(0,\ldots,0,0)$,
$\vv_1=(0,\ldots,0,1)$, $\vv_{2^n-1}=(1,\ldots,1,1)$.
The {\em  truth table} of a Boolean function $f\in\cB_n$ is the binary string of length $2^n$,
$[f(\vv_0),$ $f(\vv_1),   \ldots, f(\vv_{2^n-1})]$ (we will often disregard the commas).

If $\xx=(x_1, \ldots, x_n)$ and $\yy=(y_1, \ldots, y_n)$ are two vectors in $\BBF_2^n$
we define the {\em scalar $($or  inner$)$ product}, by
$
\xx \cdot \yy = x_1y_1 + x_2y_2 + \cdots + x_ny_n.
$
The  scalar/inner product $\xx\odot\yy$ in $\BBC\times \BBC$ is similar, with the sum over~$\BBC$.
The {\em intersection} of two vectors $\xx,\yy$ in some vector space under discussion is
$
\xx\star \yy = (x_1y_1,x_2y_2,\ldots,x_ny_n).
$
  We write $a=\Re(z), b=\Im(z)$ for the real part, respectively, imaginary part of the complex number~$z=a+b\,i  \in \mathbb{C}$, where $i^2=-1$, and $a,b\in\BBR$. Further,
$|z|=\sqrt{a^2+b^2}$ is the absolute value of $z$, and $\overline{z}=a-b\,i $ denotes
the complex conjugate of $z$.

We call a function from $\BBF_2^n$ to $\Z_q$ ($q \geq 2 $) a {\em generalized}   {\em Boolean function} on $n$ variables, and denote
the set of all generalized Boolean functions by $\cGB_{n}^q$ and, when $q=2$, by $\cB_{n}$, as previously mentioned. If $q=2^k$ for some $k\ge 1$, we can associate to any $f \in \cGB_{n}^q$ a unique sequence of Boolean functions $a_i\in \cB_{n}$ ($i=0,1,\ldots,k-1$) such that
\begin{equation*}
\label{eq0.1}
f(\xx) = a_0(\xx) + 2 a_1(\xx)+\cdots+2^{k-1} a_{k-1}(\xx), \mbox{ for all } \xx \in \BBF_2^n.
\end{equation*}
 For a Boolean function $f\in\cB_n$, we define its sign function $\hat f$ by $\hat f(\xx)=(-1)^{f(\xx)}$. In general, the sign function of $f\in \cGB_{n}^q$ is $\hat f(\xx)=\zeta_q^{f(\xx)}$, where $\zeta_q = e^{\frac{2\pi i}{q}}$, the $q$-complex root of~$1$.

For a {\it generalized Boolean function} $f:\BBF_2^n\to \Z_q$ we define the (normalized) {\it generalized Walsh-Hadamard transform} to be the complex valued function
\[ \mathcal{H}^{(q)}_f(\uu) = 2^{-n/2} \sum_{\xx\in \BBF_2^n}\zeta_q^{f(\xx)}(-1)^{\langle\uu,\xx\rangle}, \]
where $\zeta_q = e^{\frac{2\pi i}{q}}$ (for easy writing, we sometimes use $\zeta$, $\cH_f$, instead of $\zeta_q$, respectively, $\cH_f^{(q)}$,  when $q$ is fixed) and $\langle\uu,\xx\rangle$ is a nondegenerate inner product:
 for $\F_2^n$, the vector space of the $n$-tuples over $\F_2$ we use the conventional dot product $\uu\cdot\xx$ for $\langle\uu,\xx\rangle$, and
for $\F_{2^n}$, the standard inner product of $u,x\in\F_{2^n}$ is $\Tr(ux)$, where $\Tr(z)$ denotes the absolute trace of $z\in\F_{2^n}$.
For $q=2$, we obtain the usual {\it Walsh-Hadamard transform}
\[ \mathcal{W}_f(\uu) = 2^{-n/2}\sum_{\xx\in \BBF_2^n}(-1)^{f(\xx)+\langle\uu,\xx\rangle}.
\]

The sum $$\cC_{f,g}(\zz)=\sum_{\xx \in \BBF_2^n} \zeta^{f(\xx+\zz)  - g(\xx )}$$
is  the {\em crosscorrelation} of $f$ and
$g$ at $\zz \in \BBF_2^n$.
The {\em autocorrelation} of $f \in \cGB_{n}^q$ at $\uu \in \BBF_2^n$
is $\cC_{f,f}(\uu)$ above, which we denote by $\cC_f(\uu)$. 

Given a generalized Boolean function $f$, the derivative  $D_{\aa}f$ of $f$ with respect to a vector~$\aa\in\BBF_2^n$,   is the generalized Boolean function defined by
\begin{equation}
 D_{\aa}f(\xx) =  f(\xx+ \aa)-f(\xx) , \mbox{ for  all }  \xx \in \BBF_2^n.
\end{equation}

A function $f:\BBF_2^n\rightarrow\Z_q$ is called {\em generalized bent} ({\em gbent}) if $|\mathcal{H}_f(\uu)| = 2^{n/2}$ for all $\uu\in \BBF_2^n$. For descriptions of generalized bents, the reader can consult~\cite{MMMS17,smgs,txqf}.

The {\em nega--Hadamard transform} of $f\in \cGB_n^q$ at
any vector $\uu \in \BBF_2^n$ is the complex valued function:
$$
\cN_f^{(q)}(\uu) =  2^{-\frac{n}{2}}
\sum_{\xx \in \BBF_2^n} \zeta^{f(\xx) }(-1)^{ \uu \cdot \xx}\, \imath^{wt(\xx)}.$$
As customary, we may drop (some of) the superscripts in all of our notations, if convenient.
In~\cite{RPM06}, some of us considered generalized
bent criteria for Boolean functions by analyzing Boolean functions
which have flat spectrum with respect to one or more transforms
chosen from a set of unitary transforms. The transforms chosen
 are $n$-fold tensor products of
the identity mapping $\left( \begin{array}{cc}1&0\\0&1 \end{array} \right)$,
the Walsh--Hadamard transformation
$\frac{1}{\sqrt{2}}\left( \begin{array}{lr}1&1\\1&-1 \end{array} \right)$,
and the nega--Hadamard transform
$\frac{1}{\sqrt{2}}\left( \begin{array}{lr}1& \imath \\ 1 & - \imath \end{array} \right)$, where $\imath^2 = -1$.
That  choice
is motivated by  local unitary transforms that play an important
role in the structural analysis of pure $n$-qubit stabilizer quantum states.

A function $f:\BBF_2^n\rightarrow\Z_q$ is said to be (generalized) {\em negabent} if the nega--Hadamard transform is
flat in absolute value, namely $|\cN_f^{(q)}(\uu)|=1$ for all $\uu \in \BBF_2^n$.
The sum
$$\cC^n_{f,g}(\zz)=\sum_{\xx \in \BBF_2^n} \zeta^{f(\xx+\zz)  - g(\xx )}(-1)^{\xx \cdot \zz}$$
is  the {\em nega--crosscorrelation} of $f$ and
$g$ at $z$.
We define the {\em nega--autocorrelation} of $f$ at
$\uu \in \BBF_2^n$ by
\[
\cC^n_f(\uu) = \sum_{\xx \in \BBF_2^n} \zeta^{f(\xx+\zz)  - f(\xx )}(-1)^{\xx \cdot \zz}.
\]

For more on (generalized) Boolean functions, the reader can consult~\cite{CH1,PS17,MMMS17,RS19,smgs,txqf} and the references therein.

\section{The root-Hadamard transform}
Both Walsh-Hadamard and nega-Hadamard transforms can be generalized into what we will call the {\em root-Hadamard transform}:

\begin{defn}
\label{defnU}
Let $f\in\cGB_n^{2^k}$,  $\zetak$ a $2^k$-complex root of $1$ $($when convenient we drop the index$)$, $A=\{\alpha_1,\ldots,\alpha_r\}$  a set of roots of unity $\alpha_j=e^{\frac{2\pi i}{k_j}}$, $K=\{k_j\}_{1\leq j\leq r}$ and $L=\{R_s\}_{s\in K}$ be a partition  of the index set $\{0,\ldots,n-1\}=\bigsqcup_{s\in K}R_s$, $|L|=r$ $($for convenience, we index the partition by the elements of $K)$.
For $L=\{R_s\}_{s\in K}$, we let $\xx_{R_s}=(x_j)_{j\in R_s}$ and $\lambda_L(\xx)=\prod_{s\in K}\alpha_s^{wt(\xx_{R_s})}$, where  $\displaystyle wt(\xx_{R_s})=\sum_{j\in R_s}x_j\in \N$ $($observe that $\lambda_L(\00)=1$, for any partition $L$).
We define the {\em root-Hadamard transform} of $f$ at
any vector $\uu \in \BBF_2^n$ as the complex valued function:
$$
\cU_{L,A,f}^{(2^k)}(\uu) =  2^{-\frac{n}{2}}
\sum_{\xx \in \BBF_2^n}   \zeta^{f(\xx) }(-1)^{ \uu \cdot \xx}\, \lambda_L(\xx).
$$
If $f$ is a Boolean function, we let $\cT_{L,A,f}:=\cU_{L,A,f}^{(2)}$.
A function $f$ on $\cGB_n^{2^k}$ is said to be {\em root-bent} if the root-Hadamard transform is
flat in absolute value, namely $|\cU_{L,A,f}(\uu)|=1$ for all $\uu \in \BBF_2^n$.
The sum
$$\cC^p_{L,A,f,g}(\zz)=\sum_{\xx \in \BBF_2^n} \zeta^{f(\xx+\zz)  - g(\xx )} \prod_{s\in K}\mu_s^{\xx_{R_s}\odot\zz_{R_s}}$$
is  the {\em root--crosscorrelation} of $f$ and
$g$ at $z$ (recall that $\mu_s=\alpha_s^2$).
We define the {\em root--autocorrelation} of $f$ at
$\uu \in \BBF_2^n$ by
\[
\cC^p_{L,A,f}(\zz) = \sum_{\xx \in \BBF_2^n} \zeta^{f(\xx+\zz)  - f(\xx )}\prod_{s\in K}\mu_s^{\xx_{R_s}\odot\zz_{R_s}}.
\]
 When the sets $L,A$ are understood from the context, to simplify the notation, we may write $\cU_f,T_f,\cC_{f,g}^p,\cC_f^p$ in lieu of $\cU_{L,A,f}^{(2^k)}$,  $\cT_{L,A,f}^{(2^k)}$, $\cC_{L,A,f,g}^p$, $\cC_{L,A,f}^p$.
\end{defn}

\begin{example}
\label{rootbent}
Consider $A=\left\{e^{\frac{2\pi i}{4}},e^{\frac{2\pi i}{8}}\right\}$ and $L=\{(0,2),(1,3)\}$.
The generalized Boolean function $$f({\bf x})=x_1 x_2+x_2x_3+x_2x_4+x_1 x_4+x_3 x_4+2 \left(x_1 x_2+x_1 x_3+x_3 x_4\right)$$ is root-bent, that is, $|\cU_{L,A,f}({\bf u})|=1$ for all $\uu\in \mathbb{F}_2^n$.
\end{example}

This transform is related to (but more general than)  the concept of {\em consta-Hadamard} transform (see \cite{P00,S16}). First, we show that it is a proper kernel transform, so we show its invertibility.
\begin{prop}
Let $f\in\cGB_n$, $A$ be a set  of complex roots of $1$ and $\{R_k\}_{k\in K}$ be a partition of $\{0,1,\ldots,n-1\}$, indexed by the elements in $A$. Then, for any $\yy\in\BBF_2^n$, we have that
\[
\zeta^{f(\yy)}= \frac{1}{2^{\frac{n}{2}} \lambda_L(\yy)} \sum_{\ww\in\BBF_2^n}\cU_{L,A,f}(\ww)(-1)^{\yy\cdot\ww}.
\]
\end{prop}
\begin{proof}
We perform the following computation:
\begin{align*}
& 2^{-\frac{n}{2}}\sum_{\ww\in\BBF_2^n}\cU_{L,A,f}(\ww)(-1)^{\yy\cdot\ww}\\
&=
2^{-n}\sum_{\ww\in\BBF_2^n}\sum_{\xx\in\BBF_2^n}\zeta^{f(\xx) }(-1)^{ \ww \cdot \xx}\, \prod_{s\in K}\alpha_s^{wt(\xx_{R_s})}(-1)^{\yy\cdot\ww}\\
&=2^{-n}\sum_{\xx\in\BBF_2^n}\zeta^{f(\xx) }\, \prod_{s\in K}\alpha_s^{wt(\xx_{R_s})}\sum_{\ww\in\BBF_2^n}(-1)^{(\xx+\yy)\cdot\ww}\\
&=\zeta^{f(\yy) }\, \prod_{s\in K}\alpha_s^{wt(\yy_{R_s})},
\end{align*}
and the claim follows.
\end{proof}

If $\alpha$ is  a complex root of 1, we let
$\mu=\alpha^2$  (recall the scalar product $\xx\odot \zz$ is  computed over $\BBC$).
We will make use throughout of the well-known identity on binary vectors (see~\cite{MWSL})
\begin{equation}
\label{weight-sum}
\wt(\xx+  \yy)=\wt(\xx)+\wt(\yy)-2 \wt(\xx\star \yy).
\end{equation}
The following result is a collection of facts from~\cite{smgs,p12}.
\begin{prop}
\label{prop1}
We have:
\begin{itemize}
\item[$(1)$]
If $f,g\in\cGB_{n}^q$, then
\begin{equation}
\label{eq3}
\begin{split}
\sum_{\uu \in \BBF_2^n}\cC_{f,g}(\uu)(-1)^{\langle\uu,\xx\rangle} =  \cH_f(\xx)\overline{\cH_g(\xx)}, \\
\cC_{f,g}(\uu) = 2^{-n}\sum_{\xx \in \BBF_2^n}\cH_f(\xx)\overline{\cH_g(\xx)}(-1)^{\langle\uu,\xx\rangle}.
\end{split}
\end{equation}
In particular, if $f = g$, then
$\displaystyle
\cC_f(\uu) = 2^{-n}\sum_{\xx \in \BBF_2^n}|\cH_f(\xx)|^2(-1)^{\langle\uu,\xx\rangle}.
$
\item[$(2)$]
If $f, g \in \cB_n$, then the nega--crosscorrelation equals
\begin{equation*}
\cC^n_{f,g}(\zz)
= \imath^{wt(\zz)} \sum_{\uu \in \BBF_2^n} \cN_f(\uu)
                            \overline{\cN_g (\uu)}(-1)^{\uu \cdot \zz}.
\end{equation*}
\end{itemize}
\end{prop}

We now prove a result similar to Proposition~\ref{prop1} for this newly defined transform.
\begin{thm}
\label{thm:UC}
Let $f,g\in\cB_n^{2^k}$, $A$ be a set  of complex roots of $1$ and $\{R_k\}_{k\in K}$ be a partition of $\{0,1,\ldots,n-1\}$ as before. The root-crosscorrelation of $f,g$ is
\begin{align*}
\cC_{L,f,g}^p(\zz)&=\lambda_L(\zz)\sum_{\uu\in\BBF_2^n} \cU_{L,A,f}(\uu)\overline{\cU_{L,A,g}(\uu)} (-1)^{\uu\cdot\zz}.
\end{align*}
Furthermore, the root-Parseval identity holds
\[
\sum_{\uu\in\BBF_2^n} |\cU_{L,A,f}(\uu)|^2=2^n.
\]
Moreover, $f$ is root-bent if and only if $\cC_{L,A,f}^p(\uu)=0$, for all $\uu\neq \00$.
\end{thm}
\begin{proof}
Using \cite[Lemma 2.9]{PS17} and identity~\eqref{weight-sum}, we write
\begin{align*}
&\lambda_L(\zz)\sum_{\uu\in\BBF_2^n} \cU_{L,A,f}(\uu)\,\overline{\cU_{L,A,g}(\uu)} (-1)^{\uu\cdot\zz}\\
=& 2^{-n}\sum_{\xx,\yy\in\BBF_2^n} \zetak^{f(\xx)-g(\yy)} \lambda_L(\xx) \overline{\lambda_L(\yy)} \lambda_L(\zz)
\sum_{\uu\in\BBF_2^n} (-1)^{\uu\cdot(\xx +\yy +\zz)}\\
=& 2^{-n}\sum_{\xx,\yy\in\BBF_2^n} \zetak^{f(\xx)-g(\yy)}\prod_{s\in K}\alpha_s^{wt(\xx_{R_s}) -wt(\yy_{R_s}) +wt(\zz_{R_s})}
\sum_{\uu\in\BBF_2^n} (-1)^{\uu\cdot(\xx +\yy +\zz)}\\
=& \sum_{\xx,\yy\in\BBF_2^n} \zetak^{f(\xx) -g(\xx+\zz)} \prod_{s\in K}\alpha_s^{2 wt(\xx_{R_s}\star\zz_{R_s})} \\
=& \sum_{\xx\in\BBF_2^n} (-1)^{f(\xx)-g(\xx+\zz)}\prod_{s\in K}\mu_s^{\xx_{R_s}\odot\,\zz_{R_s}}=\cC_{f,g}^p(\zz).
\end{align*}

If $f=g$, then we get
\begin{align*}
\cC_{L,A,f}^p(\zz)&= \sum_{\xx\in\BBF_2^n} (-1)^{f(\xx)-f(\xx+\zz)}\prod_{s\in K}\mu_s^{\xx_{R_s}\odot\,\zz_{R_s}}\\
&=\lambda_L(\zz)\sum_{\uu\in\BBF_2^n} \cU_{L,A,f}(\uu)\,\overline{\cU_{L,A,f}(\uu)} (-1)^{\uu\cdot\zz},
\end{align*}
and by replacing $\zz=\00$, then we get the root-Parseval identity. The last claim is also implied by the previous identity.
\end{proof}
\begin{example}
 The generalized Boolean function $f({\bf x})$ in Example~\textup{\ref{rootbent}} satisfies $\cC^p_{L,A,f}({\bf u})=0$ for all ${\bf u}\neq {\bf 0}$ $($as predicted by Theorem~\textup{\ref{thm:UC}}$)$.
\end{example}

\section{Complementary sequences}

We next give a brief overview of  Golay complementary pairs (CP) (see, for example,~\cite{DJ,FJP,JP,KUS07} or the reader's preferred reference on CP). Let $\aa=\{a_i\}_{i=0}^{N-1}$ be a sequence of $\pm 1$ (bipolar) and let the {\em aperiodic autocorrelation} of $\aa$ at $k$ be defined by $\cA_\aa(k)=\sum_{i=0}^{N-k-1} a_i a_{i+k}$, $0\leq k\leq N-1$.
The {\em periodic autocorrelation} of $\aa$ at $0\leq k\leq N-1$ is defined by $\cC_\aa(k)=\sum_{i=0}^{N-1} a_i a_{i+k}$, $0\leq k\leq N-1$, where we take indices modulo~$N$. The {\em negaperiodic autocorrelation} is $\cC^n_f(\uu)=\sum_{i=0}^{2^n-1} a_ia_{i+k} (-1)^{\lfloor{(k+i)/2^n}\rfloor}$.

Two bipolar sequences $\aa,\bb$ form a ({\em Golay}) {\em complementary pair} if
\[
\cA_\aa(k)+\cA_\bb(k)=0, \text{ for } k\neq 0.
\]
We call them a $P$-complementary pair if
\[
\cC_\aa(k)+\cC_\bb(k)=0, \text{ for } k\neq 0,
\]
and $N$-complementary pair if
\[
\cC^n_\aa(k)+\cC^n_\bb(k)=0, \text{ for } k\neq 0,
\]
We associate a polynomial $A$ to the sequence $\aa$ by $A(x)=a_0+a_1x+\cdots+a_{N-1} x^{N-1}$. It is rather straightforward to show that two sequences $\aa,\bb$ (with corresponding polynomials $A,B$) form a Golay complementary pair if and only if
\begin{equation}
\label{eq:Gcp}
A(x)A(x^{-1})+B(x)B(x^{-1})=2N.
\end{equation}
Similarly, they form a $P$-complementary pair if
\begin{equation}
\label{eq:Pcp}
A(x)A(x^{-1})+B(x)B(x^{-1})\equiv 2N \pmod{x^N-1}
\end{equation}
and a $N$-complementary pair if
\begin{equation}
\label{eq:Ncp}
A(x)A(x^{-1})+B(x)B(x^{-1})\equiv 2N \pmod{x^N+1}.
\end{equation}
Let $U,V$ be the following $N\times N$ matrices, defined by
\[
U=\begin{pmatrix}
0 & 1& 0 &\cdots &0 &0\\
0 & 0& 1 &\cdots &0 &0\\
\vdots & \vdots& \vdots &\cdots &\vdots &\vdots\\
0 & 0& 0 &\cdots &0 &1\\
1 & 0& 0 &\cdots &0 &0
\end{pmatrix},\quad
V=\begin{pmatrix}
0 & 1& 0 &\cdots &0 &0\\
0 & 0& 1 &\cdots &0 &0\\
\vdots & \vdots& \vdots &\cdots &\vdots &\vdots\\
0 & 0& 0 &\cdots &0 &1\\
-1 & 0& 0 &\cdots &0 &0
\end{pmatrix}.
\]
We can quickly see that the periodic and negaperiodic autocorrelations satisfy $\cC_\aa(k)=\aa\cdot \aa U^k$ and  $\cC^n_\aa(k)=\aa\cdot \aa V^k$.
It is interesting to note that a pair of sequences is complementary if and only if it is $P$-complementary and $N$-complementary, which easily follows from the identities
\begin{align*}
\cP_\aa(k)&=\cA_\aa(k)+\cA_\aa(N-k),\\
\cC_\aa^n(k)&=\cA_\aa(k)-\cA_\aa(N-k).
\end{align*}

All of the above concepts can be extended to a set of sequences $S=\{\aa_i\}_{1\leq i\leq M}$ by imposing the sum of autocorrelations to be zero at nonzero shift, and we shall such, a {\em complementary set}) (with respect to some fixed autocorrelation). For example, the set $S$ is $P$-complementary (respectively, $N$-complementary) if $\displaystyle \sum_{i=1}^M \cA_{\aa_i}(k)=0$ (respectively, $\displaystyle \sum_{i=1}^M \cC_{\aa_i}(k)=0$), for $k\neq 0$. We shall also define the notion of {\em pairwise complementary set} $S$, by assuming that any pair within the set is complementary with respect to some autocorrelation.

The previous concepts take different forms for Boolean functions, since when we add vectors in the input of a function, we lose the circular permutation property (though, it can regarded as negacyclic permutation property. For two functions $f,g\in  \cGB_{n}^q$, we let the {\em periodic/negaperiodic correlation} of $f,g$ to be
\begin{align*}
\cC_{f,g}(\uu)&=\sum_{i=0}^{2^n-1} \zeta^{f(\vv)-g(\vv+\uu)},\\
\cC^n_{f,g}(\uu)&=\sum_{i=0}^{2^n-1}\zeta^{f(\vv)-g(\vv+\uu)} (-1)^{\uu\cdot\vv}.
\end{align*}
We say that two Boolean functions are {\em complementary} if and only if they are both $P$-complementary and $N$-complementary.

We will now define our new concept of (periodic and aperiodic) complementarity.
\begin{defn}
We say that two pairs of functions $(a_1,a_2),(b_1,b_2)$ are:
\begin{itemize}
\item[$(i)$]  {\em  $A$-crosscomplementary} if
\[
\cA_{a_1,a_2}(k)+\cA_{b_1,b_2}(k)=0, \text{ for } k\neq 0.
\]
\item[$(ii)$] {\em  $P$-crosscomplementary} if
\[
\cC_{a_1,a_2}(k)+\cC_{b_1,b_2}(k)=0, \text{ for } k\neq 0.
\]
\item[$(iii)$] {\em  $N$-crosscomplementary} if
\[
\cC^n_{a_1,a_2}(k)+\cC^n_{b_1,b_2}(k)=0, \text{ for } k\neq 0.
\]
\end{itemize}
\end{defn}

\section{Complementary pairs and components of  generalized Boolean functions}

The following lemma  from~\cite{MMMS17,smgs,txqf}, provides a relationship between the generalized Walsh-Hadamard transform and the classical transform.
Recall  the ``canonical bijection''  $\fun{\iota} {\F_2^{k-1}}{\Z_{2^{k-1}}}$, which is defined  by $\iota(\cc) = \sum_{j=0}^{k-2}c_j 2^j$ where $\cc=(c_0,c_1,\dots,c_{k-2})$.
\begin{lem}
\label{lem:H-W}
For a generalized Boolean $f\in\cGB_n^{2^k}$, $f(\xx) = a_0(\xx)+2a_1(\xx)+\cdots+ 2^{k-1}a_{k-1}(\xx)$, $a_i\in\cB_n$, we have
\begin{align}
\label{eq:relationGWaGray}
  \cH_{f} (\uu)
  &= \frac{1}{2^{k-1}} \sum_{(\cc,\dd)\in\F_ {2}^{k-1}\times\F_{2}^{k-1}} (-1)^{\cc\cdot \dd}\zetak^{\iota(\dd)}\,    \cW_{f_\cc}(\uu),
\end{align}
where $f_\cc(\xx) = c_0a_0(\xx)\+c_1a_1(\xx)\+\cdots\+c_{k-2}a_{k-2}(\xx)\+a_{k-1}(\xx)$ are the component functions of $f$.
\end{lem}
The next lemma is known and easy to show.
\begin{lem}
\label{lem:z-bit}
If $b,c$ are  bits and $z$ is a complex number, then
\begin{align*}
&2z^b=(1+(-1)^b)+(1-(-1)^b)z,\\
&\left(1+(-1)^bz\right)(-1)^{bc}
=\begin{cases}
1+z & \text{ if } b=0\\
(1-z) (-1)^c & \text{ if } b=1.
\end{cases}
\end{align*}
\end{lem}
In the next  theorem, we generalize Lemma~\ref{lem:H-W} with respect to our root-Hadamard transforms. The proof is interestingly enough, similar to the proof of  Lemma~\ref{lem:H-W}, with some appropriate changes.
\begin{thm}
\label{thm:UT}
For a generalized Boolean $f\in\cGB_n^{2^k}$, $A$ and $L$ as in Definition~\textup{\ref{defnU}}, and $f(\xx) = a_0(\xx)+2a_1(\xx)+\cdots+ 2^{k-1}a_{k-1}(\xx)$, $a_i\in\cB_n$, we have
\begin{align}
\label{eq:relationUT}
  \cU_{L,A,f} (\uu)
  &= \frac{1}{2^{k-1}} \sum_{(\cc,\dd)\in\F_ {2}^{k-1}\times\F_{2}^{k-1}} (-1)^{\cc\cdot \dd}\zetak^{\iota(\dd)}\,  \cT_{L,A,f_\cc}(\uu),
\end{align}
where $f_\cc(\xx) = c_0a_0(\xx)\+c_1a_1(\xx)\+\cdots\+c_{k-2}a_{k-2}(\xx)\+a_{k-1}(\xx)$ are the component functions of $f$.
\end{thm}
\begin{proof}
Denoting  $\displaystyle \gamma_\cc=\prod_{i=0}^{k-2} \left(1+(-1)^{c_i}\zeta_{2^{k-i}}\right)$, where $\cc=(c_0,c_1,\ldots,c_{k-2})\in\F_2^{k-1}$, we compute
\allowdisplaybreaks
 \begin{align*}
&2^{n/2}\cU_{L,A,f}(\uu)=\sum_{\xx\in\F_2^n} \zetak^{f(\xx)} (-1)^{\uu\cdot \xx} \lambda_L(\xx)
= \sum_{\xx\in\F_2^n} \zetak^{\sum_{i=0}^{k-1} f_i(\xx) 2^i} (-1)^{\uu\cdot \xx} \lambda_L(\xx)\\
&= \sum_{\xx\in\F_2^n}  (-1)^{f_{k-1}(\xx)+\uu\cdot \xx} \lambda_L(\xx)  \prod_{i=0}^{k-2} \zeta_{2^{k-i}}^{f_i(\xx)}\\
&= \frac{1}{2^{k-1}}\sum_{\xx\in\F_2^n}  (-1)^{f_{k-1}(\xx)+\uu\cdot \xx} \lambda_L(\xx) \prod_{i=0}^{k-2} \left(1+\zeta_{2^{k-i}}+\left(1-\zeta_{2^{k-i}}\right)(-1)^{f_i(\xx)} \right)\\
&=\frac{1}{2^{k-1}}\sum_{\xx\in\F_2^n}  (-1)^{f_{k-1}(\xx)+\uu\cdot \xx} \lambda_L(\xx) \\
&\qquad\cdot
\sum_{\cc\in\F_2^{k-1}} (-1)^{\sum_{i=0}^{k-2} c_i f_i(\xx)}
\prod_{i=0}^{k-2} \left(1+(-1)^{c_i}\zeta_{2^{k-i}} \right)\text{ (by Lemma~\ref{lem:z-bit})}\\
&=\frac{1}{2^{k-1}}\sum_{\cc\in\F_2^{k-1}}
\gamma_\cc \sum_{\xx\in\F_2^n}  (-1)^{f_\cc(\xx)+\uu\cdot \xx} \lambda_L(\xx)
=\frac{1}{2^{k-1}} \sum_{\cc\in\F_2^{k-1}}  \gamma_\cc  \cT_{L,A,f_\cc}(\uu).
\end{align*}
The theorem follows after easily   expressing  $\gamma_\cc$ as  a sum of powers of the complex roots of $1$ (or simply using~\cite[Lemma 5]{txqf}).
\end{proof}

The next lemma will be used later.
\begin{lem}[Inversion Lemma]
\label{lem:inv}
We let $F_\uu:\F_2^n\to\C$ be a class of complex-valued functions indexed by $\uu\in\F_2^{k-1}$. Then, for every $\cc\in\F_2^{k-1}$, we have
\[
F_\cc(\aa)=\frac{1}{2^{k-1}} \sum_{\uu,\vv\in\F_2^{k-1}} (-1)^{(\uu+\cc)\cdot\vv} F_\uu(\aa).
\]
\end{lem}
 \begin{proof}
Observe that
 \[
 \sum_{\uu,\vv\in\F_2^{k-1}} (-1)^{(\uu+\cc)\cdot\vv} F_\uu(\aa)
 =  \sum_{\uu\in\F_2^{k-1}}  F_\uu(\aa) \sum_{\vv\in\F_2^{k-1}}  (-1)^{(\uu+\cc)\cdot\vv}
 =2^{k-1}F_\cc(\aa),
 \]
 by~\cite[Lemma 2.9]{PS17}, and our result follows.
 \end{proof}

 In the spirit of the Hadamard and nega-Hadamard complementary notions, we define a root-transform complementarity notion next. For $L=\{R_j\}_{j\in K}$, a partition of $\{0,1,\ldots,n-1\}$ and $A=\{\alpha_j\}_{j\in K}$, a set of complex roots of $1$, we say that a set $S=\{g_i\}_{i=1}^M$ of functions in $\cGB_n^q$ ($q$ may be $2$) is  {\em $LA$-complementary} if
 \[
 \sum_{i=1}^M \cC_{L,A,g_i}^p(\uu)=0, \text{for all } \uu\neq 0.
 \]
 Moreover, two  tuples $S_1=(f_i)_{i=1}^M, S_2=(g_i)_{i=1}^M$ of (generalized or not) Boolean functions are {\em $LA$-crosscomplementary} if
 \[
 \sum_{i=1}^M \cC_{L,A,f_i,g_i}^p(\uu)=0, \text{for all } \uu\neq 0.
 \]

We give such an example below.
 \begin{example}
  Let $A=\left\{e^{\frac{2\pi i}{4}},e^{\frac{2\pi i}{8}}\right\}$ and $L=\{\{0,2\},\{1,3\}\}$.  Consider the generalized Boolean function $f\in\cGB_4^{4}$, defined by  $f({\bf x}) =x_1+x_2+x_3+x_4+2f_1({\bf x})$ where $f_1$ is any Boolean function in the set $S_F$ $($see Table~\textup{\ref{setSF}}$)$.  Then,
$$
\cC^p_{L,A,f}({\bf u})=\begin{cases}
 16, & {\bf u} = {\bf 0} \\
 -8i, & {\bf u} = (0, 1, 0, 1)\\
 0, & \text{otherwise}
\end{cases}.$$
Furthermore, let $g\in\cGB_4^{4}$, defined by  $g({\bf x})=x_1x_2+x_3+x_4+x_1x_4+2g_1({\bf x})$, where $g_1$ is a Boolean function in the set $S_G$ $($see Table~\textup{\ref{setSG}}$)$.  Then,
$$\cC^p_{L,A,g}({\bf u})=\begin{cases}
 16, & {\bf u} = {\bf 0} \\
 8i, & {\bf u} = (0, 1, 0, 1)\\
 0, & \text{otherwise}
\end{cases}.
$$
Clearly, $\cC^p_{L,A,f}({\bf u})+\cC^p_{L,A,g}({\bf u})=0$.  In other words, the generalized Boolean functions $f,g$ are $LA$-complementary.
  \begin{table}[H]
 \centering
 \begin{tabular}{|l|}
  \hline
 $x_1 x_2+x_3 x_2+x_1 x_4+x_3 x_4+x_4$ \\
 $x_1 x_2+x_3+x_1 x_4+x_4$ \\
 $x_1 x_2+x_3 x_2+x_3+x_1 x_4+x_3 x_4+x_4$ \\
 $x_1+x_2 x_3+x_3 x_4+x_4$ \\
 $x_2 x_1+x_4 x_1+x_1+x_2 x_3+x_3 x_4+x_4$ \\
 $x_2 x_1+x_4 x_1+x_1+x_3+x_4$ \\
 $x_2 x_1+x_4 x_1+x_1+x_2 x_3+x_3+x_3 x_4+x_4$ \\
 $x_1+x_2+x_2 x_3+x_3 x_4$ \\
 $x_2 x_1+x_4 x_1+x_1+x_2+x_2 x_3+x_3 x_4$ \\
 $x_1+x_2+x_2 x_3+x_3+x_3 x_4$ \\
 $x_2 x_1+x_4 x_1+x_1+x_2+x_3$ \\
 $x_2 x_1+x_4 x_1+x_1+x_2+x_2 x_3+x_3+x_3 x_4$ \\
 \hline
 \end{tabular}
\caption{Set of Boolean functions $S_F$.}
\label{setSF}
\end{table}

\allowdisplaybreaks
\begin{table}[H]
\centering
 \begin{tabular}{|l|}
 \hline
 $x_1 x_2+x_1 x_3 x_2+x_3 x_2+x_1 x_3+x_1 x_4+x_1 x_3 x_4$ \\
 $x_1 x_2+x_1 x_3 x_2+x_4 x_2+x_1 x_3+x_1 x_4+x_1 x_3 x_4+x_3 x_4+x_4$ \\
 $x_1 x_2+x_1 x_3 x_2+x_3 x_2+x_1 x_3+x_3+x_1 x_4+x_1 x_3 x_4$ \\
 $x_1 x_2+x_1 x_3 x_2+x_4 x_2+x_1 x_3+x_3+x_1 x_4+x_1 x_3 x_4+x_3 x_4+x_4$ \\
 $x_2 x_1+x_2 x_3 x_1+x_3 x_1+x_3 x_4 x_1+x_4 x_1+x_1+x_2 x_3$ \\
 $x_2 x_1+x_2 x_3 x_1+x_3 x_1+x_3 x_4 x_1+x_4 x_1+x_1+x_2 x_4+x_3 x_4+x_4$ \\
 $x_2 x_1+x_2 x_3 x_1+x_3 x_1+x_3 x_4 x_1+x_4 x_1+x_1+x_2 x_3+x_3$ \\
 $x_2 x_1+x_2 x_3 x_1+x_3 x_1+x_3 x_4 x_1+x_4 x_1+x_1+x_3+x_2 x_4+x_3
   x_4+x_4$ \\
 $x_2 x_3 x_1+x_3 x_1+x_3 x_4 x_1+x_1+x_2+x_2 x_4+x_3 x_4$ \\
 $x_2 x_1+x_2 x_3 x_1+x_3 x_1+x_3 x_4 x_1+x_4 x_1+x_1+x_2+x_2 x_4+x_3 x_4$ \\
 $x_2 x_3 x_1+x_3 x_1+x_3 x_4 x_1+x_1+x_2+x_2 x_3+x_4$ \\
 $x_2 x_1+x_2 x_3 x_1+x_3 x_1+x_3 x_4 x_1+x_4 x_1+x_1+x_2+x_2 x_3+x_4$ \\
 $x_2 x_3 x_1+x_3 x_1+x_3 x_4 x_1+x_1+x_2+x_3+x_2 x_4+x_3 x_4$ \\
 $x_2 x_1+x_2 x_3 x_1+x_3 x_1+x_3 x_4 x_1+x_4 x_1+x_1+x_2+x_3+x_2 x_4+x_3
   x_4$ \\
 $x_2 x_3 x_1+x_3 x_1+x_3 x_4 x_1+x_1+x_2+x_2 x_3+x_3+x_4$ \\
 $x_2 x_1+x_2 x_3 x_1+x_3 x_1+x_3 x_4 x_1+x_4 x_1+x_1+x_2+x_2 x_3+x_3+x_4 $\\
 \hline
 \end{tabular}
\caption{Set of Boolean functions $S_G$.}
\label{setSG}
\end{table}
 \end{example}

 \begin{thm}
 Let $S=\{f_i\}_{i\in I}$, $f_i\in\cGB_n^{2^k}$, be a set of generalized Boolean functions   and $A,L$ as in Definition~\textup{\ref{defnU}}. Then $S$ forms an $LA$-complementary   set if and only if for all  $\aa,\cc\in\F_2^{k-1}$,  $S_\aa=\{f_\aa\}_{f\in S}, S_\cc=\{f_\cc\}_{f\in S}$ form a binary  $LA$-crosscomplementary  set.
 \end{thm}
 \begin{proof}
 We first assume that $S $  forms an LA-complementary set. Then
 $\displaystyle \sum_{f\in S} \cC_{L,A,f}^p(\vv)=0$, for $\vv\neq \00$. From Theorem~\ref{thm:UC}, we know that $\cC_{L,A,f}^p(\vv)=\lambda_L(\vv)\sum_{\uu\in\F_2^n} |\cU_{L,A,f}(\uu)|^2 (-1)^{\vv\cdot \uu}$, and using our assumption along with Theorem~\ref{thm:UT} we obtain (we divide throughout by $\lambda_L(\vv)$)
 \allowdisplaybreaks
 \begin{align*}
 0&=\sum_{\uu\in\F_2^n}\sum_{f\in S}  |\cU_{L,A,f}(\uu)|^2  (-1)^{\vv\cdot \uu}\\
 &= \sum_{\substack{\uu\in\F_2^n\\ \aa,\bb,\cc,\dd\in\F_2^{k-1}}} (-1)^{(\aa,\cc)\cdot(\bb,\dd)} \zetak^{\iota(\bb)+\iota(\dd)}
\sum_{f\in S} \cT_{L,A,f_\aa}(\uu) \cT_{L,A,f_\cc}(\uu) (-1)^{\vv\cdot \uu}\\
 &=\sum_{ \aa,\bb,\cc,\dd\in\F_2^{k-1}} (-1)^{(\aa,\cc)\cdot(\bb,\dd)} \zetak^{\iota(\bb)+\iota(\dd)}      \sum_{\uu\in\F_2^n} \sum_{f\in S} \cT_{L,A,f_\aa}(\uu) \cT_{L,A,f_\cc}(\uu) (-1)^{\vv\cdot \uu} \\
 &=\sum_{ \aa,\bb,\cc,\dd\in\F_2^{k-1}} (-1)^{(\aa,\cc)\cdot(\bb,\dd)} \zetak^{\iota(\bb)+\iota(\dd)} \sum_{f\in S}  \cC^p_{L,A,f_\aa, f_\cc}(\vv)\\
 &=\sum_{\dd\in\F_2^{k-1}} \left( \sum_{ \aa,\bb,\cc \in\F_2^{k-1}} (-1)^{(\aa,\cc)\cdot(\bb,\dd)} \zetak^{\iota(\bb)} \sum_{f\in S}  \cC^p_{L,A,f_\aa, f_\cc}(\vv)\right)\zetak^{\iota(\dd)},
 \end{align*}
 and since $\{\zetak^{\iota(\dd)} \}_{\dd\in\F_2^{k-1}}$ is a basis of $\mathbb{Q}(\zetak)$, then, for all $\dd\in\F_2^{k-1}$,
\allowdisplaybreaks
 \begin{align*}
 0&=\sum_{ \aa,\bb,\cc \in\F_2^{k-1}} (-1)^{(\aa,\cc)\cdot(\bb,\dd)} \zetak^{\iota(\bb)} \sum_{f\in S}  \cC^p_{L,A,f_\aa, f_\cc}(\vv)\\
 &=\sum_{\bb\in\F_2^{k-1}} \left( \sum_{ \aa,\cc \in\F_2^{k-1}} (-1)^{(\aa,\cc)\cdot(\bb,\dd)}  \sum_{f\in S}  \cC^p_{L,A,f_\aa, f_\cc}(\vv)\right)\zetak^{\iota(\bb)},
 \end{align*}
 which, by the same reason as above, renders, for all $\bb,\dd\in\F_2^{k-1}$,
 \[
 \sum_{ \aa,\cc \in\F_2^{k-1}} (-1)^{(\aa,\cc)\cdot(\bb,\dd)}  \sum_{f\in S}  \cC^p_{L,A,f_\aa, f_\cc}(\vv)=0.
 \]
 Inverting the previous equation using Lemma~\ref{lem:inv}, we obtain $\displaystyle \sum_{f\in S}  \cC^p_{L,A,f_\aa, f_\cc}(\vv)=0$. The reciprocal follows easily from the identity
 \begin{align*}
 \sum_{f\in S} \cC_{L,A,f}^p(\vv)
 &= \lambda_L(\vv)\sum_{\uu\in\F_2^n}\sum_{f\in S}  |\cU_{L,A,f}(\uu)|^2  (-1)^{\vv\cdot \uu}\\
 &= \lambda_L(\vv)\sum_{ \aa,\bb,\cc,\dd\in\F_2^{k-1}} (-1)^{(\aa,\cc)\cdot(\bb,\dd)} \zetak^{\iota(\bb)+\iota(\dd)} \sum_{f\in S} \cC^p_{L,A,f_\aa, f_\cc}(\vv).
 \end{align*}
The theorem is shown.
  \end{proof}

\begin{cor}
Let $\{f,g\}$ be two generalized Boolean functions in $\cGB_n^{2^k}$. Then  $\{f,g\}$ forms an $P$-complementary $($respectively, $N$-complementary$)$  pair if and only if for all $\aa,\cc\in\F_2^{k-1}$,  $\{f_\aa,f_\cc\}$ and  $\{g_\aa,g_\cc\}$ form a  $($binary$)$ $P$-crosscomplementary $($respectively, $N$-crosscomplementary$)$ pair.
\end{cor}

\section{Transforms and complementary constructions}

It is well known~\cite{MA07} that a Boolean function $f$ has flat spectrum with respect to the nega-Hadamard spectrum if $f+s_2$ (where $s_2$ is the  quadratic elementary symmetric polynomial) has flat spectrum with respect to the Walsh-Hadamard transform. Thus, it is natural to ask whether some connection exists between the Walsh-Hadamard and the nega-Hadamard transforms defined on generalized  Boolean functions. To that effect, we show the first claim of our next result (see~\cite{AKMT} for the particular case of $k=1$).
We can also   relate  the root-Hadamard transform and the generalized Hadamard transform.

We let $A=\{\alpha_1,\ldots,\alpha_r\}$  be a set of roots of unity $\alpha_j=e^{\frac{2\pi i}{k_j}}$ such that $k_j=2^{m_j},\,m_j\leq k$, $K=\{k_j\}_{1\leq j\leq r}$ and $L=\{R_s\}_{s\in K}$ be a partition  of the index set $\{0,\ldots,n-1\}=\bigsqcup_{s\in K}R_s$, $|L|=|K|=r$. For $J\subseteq K$, we let $A_J$ be $A$ with $\alpha_s$ replaced by $1$ for all $s\in J$. Let $s_1(\xx)=\displaystyle\bigoplus_{j=1}^n x_j$, $s_2(\xx)=\displaystyle\bigoplus_{1\leq j<k\leq n} x_jx_k$, and in general $s_t(\xx)=\displaystyle\bigoplus_{1\leq j_1<\ldots<j_t\leq n} x_{j_1}\cdots x_{j_1}$ be the symmetric polynomials of degree 1, 2, $t$, respectively, all reduced modulo 2. We will use here Lemma 5 from~\cite{S16}:

\begin{lem}[\cite{S16}] \label{weight} Let $\xx\in \V_n$. Then,
\begin{align*}
 wt(\xx) \pmod 4 &=s_1(\xx)+2s_2(\xx) \\
 wt(\xx) \pmod {2^k}&=wt(\xx) \pmod {2^{k-1}}+2^{k-1}s_{2^{k-1}}(\xx)
 \end{align*}
\end{lem}

\begin{thm}
 \label{relationships}
 We have:
\begin{itemize}
\item[$(i)$] Let $n\geq 1$, $k\geq 1$, $f\in\cGB_n^{2^k}$ and $g\in\cGB_n^{2^{k+1}}$ defined by $g(\xx)=2f(\xx)+2^{k-1} s_1(\xx)+2^k s_2(\xx)$ $($the sum is taken modulo $2^{k+1})$.
 Then, $\cN_f^{(2^k)}(\uu)=\cH_g^{(2^{k+1})}(\uu)$, for all $\uu\in\F_2^n$.
\item[$(ii)$]
Let $n\geq 1$, $k\geq 1$, $f,\,h_K,h_J\in\cGB_n^{2^k}$  defined by $h_K(\xx)=f(\xx)-\sum_{s\in K}2^{k-m_s} \sum_{j=0}^{s-1}s_{2^j}(\xx_{R_s})$, $h_J(\xx)=f(\xx)-\sum_{s\in  J}2^{k-m_s} \sum_{j=0}^{s-1}s_{2^j}(\xx_{R_s})$ $($the sum is taken modulo $2^{k})$, where $\xx_{R_s}$ is the restriction of $\xx$ to the indices in $R_s$. Then,
\[
\cU_{L,A,h_K}^{(2^k)}(\uu)=\cH_f^{(2^{k})}(\uu) \text { and } \cU_{L,A,h_J}^{(2^k)}(\uu)=\cU_{L,A_J,f}^{(2^k)}(\uu),\text { for all } \uu\in\F_2^n.
\]
\end{itemize}
\end{thm}
\begin{proof}
We compute
\allowdisplaybreaks
\begin{align*}
\cH_g^{(2^{k+1})}(\uu)
&=2^{-n/2} \sum_{\xx\in\F_2^n} \zeta_{2^{k+1}}^{g(\xx)} (-1)^{\uu\cdot \xx}\\
&= 2^{-n/2}  \sum_{\xx\in\F_2^n} \zeta_{2^k}^{f(\xx)}  (-1)^{\uu\cdot \xx} i^{s_1(\xx)+2 s_2(\xx)}\\
&= 2^{-n/2}  \sum_{\xx\in\F_2^n} \zeta_{2^k}^{f(\xx)}  (-1)^{\uu\cdot \xx} i^{wt(\xx)}=\cN_f^{(2^k)}(\uu).
\end{align*}

The two claims of $(ii)$ are similar so we will show only the first part. Note that, by Lemma~\ref{weight}, $wt(\xx)\pmod {2^s}=\sum_{j=0}^{s-1}s_{2^j}(\xx)$. We compute
\begin{align*}
\cU_{L,A,h_K}^{(2^k)}(\uu)
&=2^{-n/2} \sum_{\xx\in\F_2^n} \zeta_{2^{k}}^{g(\xx)} (-1)^{\uu\cdot \xx}\, \lambda_L(\xx)\\
&= 2^{-n/2}  \sum_{\xx\in\F_2^n} \zeta_{2^k}^{f(\xx)-\sum_{s\in K}2^{k-m_s} \sum_{j=0}^{s-1}s_{2^j}(\xx_{R_s})}  (-1)^{\uu\cdot \xx} \prod_{s\in K}\alpha_s^{wt(\xx_{R_s})}\\
&= 2^{-n/2}  \sum_{\xx\in\F_2^n} \zeta_{2^k}^{f(\xx)} \prod_{s\in K}\alpha_s^{-wt(\xx_{R_s})} (-1)^{\uu\cdot \xx} \prod_{s\in K}\alpha_s^{wt(\xx_{R_s})}\\
&= 2^{-n/2}  \sum_{\xx\in\F_2^n} \zeta_{2^k}^{f(\xx)}  (-1)^{\uu\cdot \xx} =\cH_f^{(2^{k})}(\uu),
\end{align*}
and the theorem is shown.
\end{proof}

\def\supp{\rm supp}

Two of us introduced the next concept in~\cite{RS19}. We call a function $f\in\cGB_n^{q}$ a {\em landscape} function if there exist $t\geq 1$, $m_i\in\N_0,\ell_i\in2\N_0+1$, $1\leq i\leq t$, such that
\[
\{|\cH_f(\uu)|\}_{\uu\in\supp(\cH_f)}=\{2^{\frac{m_1}{2}}\ell_1,\ldots,2^{\frac{m_t}{2}}\ell_t\}.
\]
 We call the set of pairs $\{(m_1,\ell_1),(m_2,\ell_2),\ldots\}$, the {\em levels}  of $f$, and  $t+1$ (if $0$ belongs to the Walsh-Hadamard spectrum), or $t$ (if $0$ is not in the spectrum) the {\em length} of   $f$.

 We can deduce this corollary (the second claim of the next result was also shown in~\cite{AKMT}).
\begin{cor}
Let $f,h\in\cB_n$, where $h(\xx)=f(\xx)+s_2(\xx)$. If $n$ is even, then $f$ is negabent if and only if $h$ is bent. Furthermore,  $f$ is negaplateaued if and only if $h$ is plateaued.
In general, $f$ is negalandscape $($defined as above, via the nega-Hadamard transform$)$ if and only if $h$ is landscape.
\end{cor}
\begin{proof}
By Theorem~\ref{relationships}\,$(i)$, for $k=1$, we have that if $f\in\cB_n$ and $g\in\cGB_n^{4}$ defined by $g(\xx)=2f(\xx)+ s_1(\xx)+2 s_2(\xx)$, 
 then, $\cN_f^{(2)}(\uu)=\cH_g^{(4)}(\uu)$, for all $\uu\in\F_2^n$. Since the decomposition of $g$ is $g(\xx)=a_0(\xx)+2 a_1(\xx)$, where $a_0(\xx)=s_1(\xx)$, and $a_1(\xx)=f(\xx)+  s_2(\xx)$, this implies, by \cite{MMMS17}, that, when $n$ is even, $g$ is gbent if and only if $a_1$ and $a_0+  a_1$ are bent Boolean functions. This implies that $g$ is gbent if and only if $f+  s_2$ and $s_1+  f+  s_2$ are bent Boolean functions. Since $s_1$ is a linear Boolean function, this implies that $f$ is negabent if and only if $h=f+  s_2$ is bent, giving yet another proof to this known result~\cite{MA07}.

Further, by Corollary 1 of \cite{RS19}, we see that, if $g:\F_2^n\rightarrow\Z_{2^k}$  is a  function given by
$g(\xx) = a_0(\xx)+2a_1(\xx)+\cdots+2^{k-1} a_{k-1}$, and  $s\geq 0$ is an integer, then, $g$ is $s$-gplateaued if and only if,  for each $\cc\in\F_2^{k-1}$, the Boolean function $g_\cc$
defined as $ g_\cc(\xx) = \cc\cdot (a_0(\xx),\ldots,a_{k-2}(\xx))+  a_{k-1}(\xx)$
is an $s$-plateaued (if $n+s$ is even), respectively, an $(s+1)$-plateaued function (if $n+s$ is odd), with some extra conditions on the Walsh-Hadamard coefficients. In particular, taking $k=1$, this implies that $f$ is negaplateaued if and only if $f+  s_1+  s_2$ is plateaued (no extra conditions on the Walsh-Hadamard coefficients), which again implies that $f+  s_2$ is plateaued.

Using Theorem 3.2 of \cite{RS19}, this argument can be also extended to landscape functions, in a similar way as in the plateaued case.
\end{proof}


\section{Further comments}

In this paper we defined a class of transforms which generalize many others, like generalizes the Walsh-Hadamard, nega-Hadamard, $2^k$-Hadamard~\cite{S16}, consta-Hadamard~\cite{P00} and $HN$-transforms. For generalized Boolean functions, we describe its behavior on the binary components. Further, we define a notion of complementarity (in the spirit of the Golay sequences) with respect to this transform and furthermore, we describe the complementarity of a generalized Boolean set with respect to the binary components of the elements of that set. Some concrete examples are provided.

There are many questions one can ask on the new transforms. For example, it would be interesting to provide more constructions of root-bent and more generally, root-plateaued functions (surely, Theorem~\ref{relationships}  may help). Certainly, finding connections between these transforms, their values, and (relative) difference sets would be quite interesting, as well.

\end{document}